\documentclass[a4paper,11pt]{article}
\usepackage{fullpage}
\usepackage{times}  
\usepackage{helvet}  
\usepackage{courier}  
\usepackage[hyphens]{url}  
\usepackage{graphicx} 
\usepackage{caption} 
\usepackage{algorithm}
\usepackage{algorithmic}
\usepackage{boxedminipage,tcolorbox,thm-restate}

\usepackage{newfloat}
\usepackage{listings}
\DeclareCaptionStyle{ruled}{labelfont=normalfont,labelsep=colon,strut=off} 
\floatstyle{ruled}
\newfloat{listing}{tb}{lst}{}
\floatname{listing}{Listing}
\usepackage{xspace}
\usepackage{amsmath}
\usepackage{amsthm}
\usepackage{amsfonts}
\usepackage{amssymb}
\usepackage{cleveref}
\usepackage{tikz}
\usepackage{multirow}
\usepackage{array}

\newcommand{\setitemtypes}{\ensuremath{T}}

\newcommand{\extensionproblem}{\textsc{Envy-Free Allocation Extension}\xspace}
\newcommand{\REFAE}{\textup{\textsc{Refae}}\xspace}
\newcommand{\FEFAE}{\textup{\textsc{Fefae}}\xspace}

\newcommand{\indset}{\textsc{Independent Set}\xspace}

\newcommand{\multicolorclique}{\textsc{Multicolored Clique}\xspace}

\renewcommand{\P}{\textsf{P}\xspace}
\newcommand{\NP}{\textsf{NP}\xspace}
\newcommand{\NPh}{\NP-hard\xspace}

\newcommand{\NPc}{\NP-complete\xspace}

\newcommand{\FPT}{\textsf{FPT}\xspace}
\newcommand{\XP}{\textsf{XP}\xspace}
\newcommand{\W}[1][1]{\textsf{W[#1]}\xspace}
\newcommand{\Wh}[1][1]{\W[#1]-hard\xspace}
\newcommand{\Whness}[1][1]{\W[#1]-hardness\xspace}

\newcommand{\bigoh}{\mathcal{O}}

\newcommand{\poly}{\ensuremath{\mathtt{poly}}\xspace}

\newcommand{\agenttypes}{\ensuremath{n_t}\xspace}
\newcommand{\itemtypes}{\ensuremath{m_t}\xspace}

\newtheorem{proposition}{Proposition}
\newtheorem{observation}{Observation}

\newtheorem{definition}{Definition}
\newtheorem{lemma}{Lemma}

\newif\iflong
\newif\ifshort

\longtrue

\iflong
\else
\shorttrue
\fi

\usepackage{graphicx} 

\usepackage{algorithm}
\usepackage{algorithmic}
\usepackage{boxedminipage,tcolorbox,thm-restate}

\usepackage[switch]{lineno}

\usepackage{xspace}
\usepackage{amsmath}
\usepackage{amsfonts}
\usepackage{amssymb}
\usepackage{cleveref}
\usepackage{tikz}
\usepackage{multirow}
\usepackage{array}

\title{The Complexity of Extending Fair Allocations of Indivisible Goods}

\author{Argyrios Deligkas\textsuperscript{1}, Eduard Eiben\textsuperscript{1}, Robert Ganian\textsuperscript{2},\\ Tiger-Lily Goldsmith\textsuperscript{1} and Stavros D. Ioannidis\textsuperscript{1} \\[1em]
    \textsuperscript{1}Department of Computer Science, Royal Holloway, University of London \\
    \textsuperscript{2}Algorithms and Complexity group, TU Wien, Austria 
    }
\date{}

\begin{document}

\maketitle

\begin{abstract}
We initiate the study of computing envy-free allocations of indivisible items in the extension setting, i.e., when some part of the allocation is fixed and the task is to allocate the remaining items. Given the known \NP-hardness of the problem, we investigate whether---and under which conditions---one can obtain fixed-parameter algorithms for computing a solution in settings where most of the allocation is already fixed. Our results provide a broad complexity-theoretic classification of the problem which includes: (a) fixed-parameter algorithms tailored to settings with few distinct types of agents or items; (b) lower bounds which exclude the generalization of these positive results to more general settings. We conclude by showing that---unlike when computing allocations from scratch---the non-algorithmic question of whether more relaxed EFX allocations exist can be completely resolved in the extension setting. 
\end{abstract}

\section{Introduction}
Finding a ``fair'' allocation of indivisible items or resources to a provided set of agents, each with their own preferences \iflong over the resources\fi, is one of the central tasks arising in the area of computational social choice. The arguably most classical and established notion of fairness used in these settings is \emph{envy-freeness}, where we ask for an allocation $\pi:M\rightarrow N$ from the set $M$ of items to the set $N$ of agents such that for each pair of agents $i,j\in N$, $i$ prefers $\pi^{-1}(i)$ to $\pi^{-1}(j)$. In other words, no agent envies another agent. It is well known that an envy-free allocation need not exist, and in fact, determining whether one exists is \NP-complete~\cite{BouveretL08}.

The aforementioned intractability has led to a flurry of research to circumvent this issue. One approach that has been proposed is to consider ``less restrictive'' versions of envy-freeness instead, with the aim of not only ensuring that the computation of an assignment is tractable but---even more desirably---that one always exists. Notably, it is known that one can always compute, in polynomial time, an allocation which is \emph{envy-free up to one item} (\emph{EF1})~\cite{amanatidis2020multiple,budish2011combinatorial,0001KSY24}. However, such allocations may sometimes be considered very far from ``fair''. A recently proposed intermediate notion between envy-free and EF1 allocations is \emph{envy-free up to any item} (\emph{EFX})~\cite{CaragiannisKMPS19,FeldmanMP24}, but it is not known whether EFX allocations always exist and their polynomial-time computability remains open as well.

Another notable approach to tackling  the problem of computing envy-free allocations is to identify precise conditions under which it can be solved efficiently. This is typically done by investigating the problem through the lens of \emph{parameterized complexity}~\cite{DowneyF13,CyganFKLMPPS15}---a refinement of the classical complexity paradigm where inputs are analyzed not only with respect to their size $n$, but also to a numerical parameter $k$ which measures some well-defined quantity. While we cannot hope to obtain an $n^{\bigoh(1)}$ algorithm for computing envy-free allocations in general, in this setting one aims to design so-called \emph{fixed-parameter tractable} (or \FPT) algorithms for the problem, that is, algorithms running in time $f(k)\cdot n^{\bigoh(1)}$ for some computable function $f$. The simplest parameters considered in past works on fair division include, for example, the number of agents or items; however, such parameterizations place strong restrictions on the input instances, and so more recent works have focused on developing algorithms parameterized by the number of \emph{agent types} (where two agents have the same type if they have identical preferences over the items), or the number of \emph{item types} (where two items have the same type if they are valued the same by every agent)~\cite{DeligkasEGHO21,EibenGHO23,0004R23}, see also~\cite{BranzeiLM16,GanianOR19}.

While the aforementioned approaches have by now provided a fairly detailed understanding of computing an envy-free allocation from scratch, in this article we turn to the problem of extending a partial allocation---that is, computing an envy-free allocation when part (or even most) of the allocation is already fixed. This \emph{envy-free allocation extension} problem arises naturally whenever one needs to deal with resources that have already been assigned or must be assigned to certain agents---consider, e.g., the case where a few new items are made available after an allocation has been fixed and we are not allowed to take items from agents, or the setting where most employees in a company already have fixed tasks, but we need to distribute a set of new tasks to recently hired employees. It is worth noting that while the complexity of extending partial solutions has been extensively researched in settings as diverse as data completion~\cite{GanianKOS18,KoanaFN21,EibenGKOS23,KoanaFN23} and graph drawing~\cite{AngeliniBFJKPR15,ArroyoDP19,EibenGHKN20,BhoreGKMN23}, 
this has not yet been systematically studied in classical resource allocation in spite of the many situations in which one may need to deal with some items being pre-assigned.\footnote
{In parallel to this article, a separate work investigating other aspects of extending fair allocations has very recently been published~\cite{igarashi2024fair}. In this work, the authors study the complexity, from the \P vs \NP point of view, of extending fair allocation with respect to relaxed notions of envy-freeness such as EF1, MMS and PROP1.}

\subsection{Our Contribution}
We study the envy-free allocation extension problem in the classical setting of additive utilities, as outlined below.

      \noindent
      \begin{center}
      \begin{tcolorbox}[title=~\textsc{Envy-Free Allocation Extension},left=-1.5mm,top=0mm,bottom=0mm,right=0mm,boxsep=1mm]
        \begin{tabular}{p{.15\textwidth}p{.8\textwidth}}
            \textbf{Input:} & \parbox[t]{.74\textwidth}{A set $M$ of indivisible items, a set $N$ of $n$ agents with additive valuations, and a partial allocation $\gamma:M'\rightarrow N$ of some subset $M'\subseteq M$ of items.}\\[20pt]
            \textbf{Question:} & \parbox[t]{.74\textwidth}{Does there exist an extension allocation $\pi:(M\setminus M')\rightarrow N$ of the remaining (``open'') items such that $\gamma \cup \pi$ is envy-free?}\\
        \end{tabular}
      \end{tcolorbox}
      \end{center}     

Notice that \textsc{Envy-Free Allocation Extension} coincides with the usual problem of obtaining an envy-free allocation from scratch in the special case where $\gamma=\emptyset$, and hence is necessarily \NP-hard. However, this intractability does not properly reflect typical usage scenarios of the problem: in many cases of interest, one may be dealing with allocations $\gamma$ which are ``almost complete''. Hence, we turn to the aforementioned parameterized paradigm and ask which parameterizations of the partial assignment $\gamma$ allow us to achieve fixed-parameter traceability, i.e., obtain a fixed-parameter algorithm. 

We begin our investigation by considering the complexity of the problem with the number $k=|M\setminus M'|$ of yet-to-be-allocated or open items as the parameter. 
It is worth noting that \textsc{Envy-Free Allocation Extension} can be solved via a trivial $n^{\bigoh(k)}$ algorithm (a so-called \XP \emph{algorithm}), and that computing an envy-free allocation from scratch is trivially fixed-parameter tractable w.r.t.\ $k$---hence, one could have expected \textsc{Envy-Free Allocation Extension} to be fixed-parameter tractable w.r.t.\ $k$ as well. However it turns out that this is not the case: we will later show (in Theorem~\ref{thm:kmtwhard}) that the problem is \W[1]-hard when parameterized by $k$, which excludes fixed-parameter tractability under well-established complexity assumptions.

Unsatisfied by this outcome, we ask whether one can at least obtain a fixed-parameter algorithm under some mild restrictions on the input instance. For example instances with a bounded number of item or agent types arise naturally in large-scale models: one often aggregates similar items into categories (depending on the setting, these may be, e.g., ``cars'' and ``houses'', or ``PhD students'' and ``Postdocs''), and the inherent limitations of how preferences are collected may result in many agents being represented as having the same preferences. Inspired by related works on resource allocation~\cite{DeligkasEGHO21,EibenGHO23,0004R23}, we consider additionally parameterizing by the number $n_t$ of the agent types or the number $m_t$ of item types. As our first positive result, we obtain an algorithm when we bound the number of agent types.

\begin{restatable}{theorem}{kntfpt}\label{thm:kntfpt}
  \extensionproblem is fixed-parameter tractable when parameterized by the number $k$ of open items plus the number \agenttypes of agent types.
\end{restatable}

Turning to the analogous question for item types, we show that the situation is entirely different: allocating $k$ open items is intractable even when there are only a few item types. 

\begin{restatable}{theorem}{kmtwhard}\label{thm:kmtwhard}
  \extensionproblem is \textup{\W[1]-hard} when parameterized by the number $k$ of open items plus the number \itemtypes of item types.
\end{restatable}

The proof of \Cref{thm:kntfpt} relies on a combination of branching techniques with insights into the structure of hypothetical solutions; while it is not trivial, in this case it is the lower bound---\Cref{thm:kmtwhard}---which presents a much more surprising result. 
Indeed, at first glance the problem seems far from intractable: we only need to assign $k$ items to a set of agents, and for each open item $i$, we can easily test whether assigning it to an agent $j$ would make any of the other agents in $N$ envy $j$ w.r.t.\ their current items. 
The only caveat is that we do not know which of the other agents will receive the remaining open items---and in spite of there being only at most $k$ open items in total, we show that the uncertainty of how these will be allocated can be used to obtain a highly non-trivial reduction from the classical \W[1]-hard \textsc{Multicolored Clique} problem~\cite{DowneyF13}.

In the second part of our article, we take aim at the setting where many items may be open, but we are only allowed to allocate these items to at most $p$-many agents. This is precisely the situation that arises when, e.g., allocating tasks to a few new employees, or needing to distribute new items to individuals while adhering to constraints on shipping costs, or how many of the individuals should be contacted.
Formally, this is captured by asking for an allocation $\pi$ which assigns all open items to a set $S$ of at most $p$-many agents. We will distinguish between whether the set $S$ of potential recipients is provided in the input (in which case we speak of \textsc{Restricted Envy-Free Allocation Extension}, or simply \REFAE) or can be selected along with $\pi$ (giving rise to \textsc{Free Envy-Free Allocation Extension}, or simply \FEFAE). 
Distinguishing these two variants of our problem is not only necessary to differentiate between distinct usage scenarios, but will also turn out to have an impact on the problem's complexity.

We note that while both variants are in \NP, attempting to solve either of them when parameterized by $p$ (the number of agents that receive items) alone is entirely hopeless---and this holds even in the strong setting, i.e., with unary-encoded valuations.

\begin{restatable}{theorem}{kappahardtwo}\label{thm:kappahardtwo}
  Both \REFAE and \FEFAE are strongly \textup{\NP-hard} even when $p\leq 2$. 
\end{restatable}
This once again raises the question of whether one can use a bound on $n_t$ or $m_t$ to obtain efficient algorithms. 
Given \Cref{thm:kntfpt} and \Cref{thm:kmtwhard}, one might expect that agent types would once again allow us to obtain fixed-parameter algorithms in this case. 
And yet, we show that when dealing with $p$ the situation is reversed: unexpectedly, here it is the number $m_t$ of item types that provides the key to tractability. Indeed, one can establish strong \W[1]-hardness as well as weak para\NP-hardness of \REFAE and \FEFAE when parameterized by the total number of agents in the instance (which upper-bounds both $n_t$ and $p$) directly from existing results in the literature~\cite{BliemBN16}. Before proceeding, we at least show that the former lower bound is tight: the problems admit \XP-algorithms in the case where the utilities are not encoded in binary.

\begin{restatable}{theorem}{kappantxp}\label{thm:kappantxp}
  If the valuations are encoded in unary, both \REFAE and \FEFAE are in \textup{\XP} when parameterized by the number $p$ of recipients plus the number $n_t$ of agent types.
\end{restatable}

Turning to the setting with few item types, we provide a fixed-parameter algorithm for the restricted case (i.e., when the set $S$ of recipients is already given) and an \XP algorithm for the free case (when the task is to also find $S$ on its own). Here, the latter result is the best one could hope for at this point, as fixed-parameter tractability is immediately excluded by \Cref{thm:kmtwhard}. In contrast to \Cref{thm:kappantxp}, both of these algorithms can be applied regardless of whether the utilities are encoded in unary or binary.

\begin{restatable}{theorem}{kappamtfpt}\label{thm:kappamtfpt}
  When parameterized by $p$ plus the number $m_t$ of item types, \REFAE and \FEFAE are fixed-parameter tractable and \textup{\XP}-tractable, respectively.
\end{restatable}

An overview of our complete complexity-theoretic classification of \extensionproblem under the considered restrictions is provided in~\Cref{tab:results}. 

While not part of our main technical contributions, in the final section of the article, we also make a few observations about the behavior of relaxed variants of envy-freeness in the extension setting that may be of general interest to the community. In particular, we prove that the known result guaranteeing the existence of an EF1 allocation can be strengthened to also hold when extending \emph{any} provided partial allocation. Moreover, this result is tight: there exist partial EFX allocations which cannot be extended to a full EF1 allocation. Similarly, while the question of whether an EFX allocation ``from scratch'' is guaranteed to exist remains one of the main open questions in the field, we show that there exist partial envy-free allocations which cannot be extended to a full EFX allocation. 

\begin{table*}[t]
\centering
\begin{tabular}{cc||cc}
 & \multirow{2}{*}{Few items $k$} & \multicolumn{2}{c}{Few Recipients $p$}    \\ 
 &  & \REFAE & \FEFAE \\ \hline \hline
  \multicolumn{1}{l|}{\rule{0pt}{2ex}} & \Wh[1] (\Cref{thm:kmtwhard})  & \multicolumn{2}{c}{\NPh, $p \leq 2$ (\Cref{thm:kappahardtwo})}  \\ \hline 
 \multicolumn{1}{c|}{\rule{0pt}{2ex} Agent types \agenttypes}& \FPT (\Cref{thm:kntfpt})& 
 \multicolumn{2}{c}{\XP, unary valuations (\Cref{thm:kappantxp})}
 \\ 
 \hline
\multicolumn{1}{c|}{\rule{0pt}{2ex} Item types \itemtypes}& \hspace{0.5cm}\Wh[1] (\Cref{thm:kmtwhard}) \hspace{0.5cm}~& \multicolumn{1}{c|}{\hspace{0.5cm}\FPT (\Cref{thm:kappamtfpt})\hspace{0.5cm}~} & \hspace{0.5cm}\XP (\Cref{thm:kappamtfpt}) \hspace{0.5cm}~\\ \hline
\end{tabular}
\caption{An overview of our results. 
The \Whness[1] results of \Cref{thm:kmtwhard} are complemented by the naive, brute force, \XP algorithms.
The \XP algorithms for \REFAE and \FEFAE when parameterized by \agenttypes are essentially tight, since matching strong \Whness[1] and weak para\NP-hardness lower bounds follow from the literature~\cite{BliemBN16}.}
\label{tab:results}
\end{table*}

\section{Preliminaries}\label{sec:prelims}
For an integer $\ell$, we use $[\ell]$ as shorthand for the set $\{1,\dots,\ell\}$. We also use the standard $\bigoh^*$ notation to suppress polynomial factors of the input size in the running time.

Our instances include a 
 set of indivisible items $M=\{a_1, \ldots, a_m\}$ and a set of $n$ agents $N=[n]$. 
Every agent \iflong $i \in N$ \fi \ifshort $i$ \fi  has an {\em additive valuation} function $v_i$ that assigns a non-negative value $v_i(a)$ for every item $a \in M$ and for every subset, or {\em bundle}, of items $B \subseteq M$ we denote $v_i(B) := \sum_{j\in B} v_i(a_j)$. 
If $v_i(a)=v_j(a)$ for every $a \in M$, then we say that agents $i$ and $j$ are of the same type. 
We will use \agenttypes to denote the number of different {\em agent types} in $N$.
If for two items $a$ and $a'$ we have that $v_i(a) = v_i(a')$ for every agent $i$, then we say that $a$ and $a'$ are of the same type.
We will use \itemtypes to denote the number of different {\em item types} in $M$.

\paragraph{\bf Partial and Extended Allocations.} 
We will assume that $M$ is partitioned into a set of {\em given} items $M'$ and a set of {\em open} items $A$. Formally, $M = M' \cup A$ and $M' \cap A = \emptyset$. 
A {\em partial allocation} $\gamma = (\gamma_1, \gamma_2, \ldots, \gamma_n)$ is a partition of $M'$ into $n$ (potentially empty) sets, where $\gamma_i$ is the bundle of given items to agent $i$.
An allocation of open items $\pi = (\pi_1, \pi_2, \ldots, \pi_n)$ partitions $A$ into $n$ (again, potentially empty) sets, where agent $i$ gets bundle $\pi_i$. Given a partial allocation $\gamma$ and an allocation $\pi$, we get an {\em extended allocation} $\gamma \cup \pi$, where agent $i$ gets allocated bundle $\gamma_i \cup \pi_i$.

\paragraph{Fairness Concepts.}
We will consider three different fairness notions.
An extended allocation $\gamma \cup \pi$ is said to be:
\begin{itemize}
    \item {\em envy-free}, denoted EF, if for any pair of agents $i$ and $j$ it holds that $v_i(\gamma_i \cup \pi_i) \geq v_i(\gamma_j \cup \pi_j)$;
    \item {\em envy-free up to any item}, denoted EFX, if for any pair of agents $i$ and $j$ and any item $a \in \gamma_j \cup \pi_j$ it holds that $v_i(\gamma_i \cup \pi_i) \geq v_i(\gamma_j \cup \pi_j \setminus a)$;
    \item {\em envy-free up to one item}, denoted EF1, if for any pair of agents $i$ and $j$ there exists an item $a \in \gamma_j \cup \pi_j$ such that $v_i(\gamma_i \cup \pi_i) \geq v_i(\gamma_j \cup \pi_j \setminus a)$.
\end{itemize}
Observe that EF1 is a more relaxed fairness notion compared to EFX, which in turn is more relaxed than EF.

\paragraph{\bf Fair Allocation Extension Problems.}
We are interested in the computation of an allocation $\pi$ of open items, such that the extended allocation is fair according to some of the above-mentioned criteria; we define the problems for envy-free solutions, but they could naturally be extended for EF1 and EFX. 
The input to each version of our problem is a set of agents with their valuations and a partial allocation $\gamma$. 
In the first version, termed \textsc{Envy-Free Allocation Extension} and already introduced in the Introduction, we do not constrain the extension allocation in any way.
In the other two versions of the problem, the set of recipients is restricted in some way. 
One option is to restrict the set of agents that are allowed to receive an open item under $\pi$, which for brevity we denote \REFAE.
      \begin{center}
      \begin{tcolorbox}[title=~\textsc{Restricted Envy-Free Allocation Extension (REFAE)},left=-1.5mm,top=0mm,bottom=0mm,right=0mm,boxsep=1mm]
        \begin{tabular}{p{.15\textwidth}p{.8\textwidth}}
            \textbf{Input:} & \parbox[t]{.74\textwidth}{A set $M$ of indivisible items, a set $N$ of $n$ agents, a partial allocation \newline $\gamma:M'\rightarrow N$ of some subset $M'\subseteq M$ of items, and a set $S \subseteq N$.}\\[20pt]
            \textbf{Question:} & \parbox[t]{.74\textwidth}{Does there exist an allocation 
            \iflong $\pi:(M\setminus M')\rightarrow N$ \fi
            \ifshort $\pi$ \fi
            of the open items such that: (a) $\pi_i = \emptyset$ for every $i \notin S$; (b) $\gamma \cup \pi$ is envy-free?}\\
        \end{tabular}
      \end{tcolorbox}
      \end{center}     

\noindent
A different option is to restrict just the number of recipients. We term this problem \FEFAE.

\noindent
      \begin{center}
      \begin{tcolorbox}[title=~\textsc{Free Envy-Free Allocation Extension (FEFAE)},left=-1.5mm,top=0mm,bottom=0mm,right=0mm,boxsep=1mm]
        \begin{tabular}{p{.15\textwidth}p{.8\textwidth}}
            \textbf{Input:} & \parbox[t]{.74\textwidth}{A set $M$ of indivisible items, a set $N$ of $n$ agents, a partial allocation \newline $\gamma:M'\rightarrow N$ of some subset $M'\subseteq M$ of items, and $p \in \mathbb{N}$.}\\[20pt]
            \textbf{Question:} & \parbox[t]{.74\textwidth}{Does there exist an allocation 
            \iflong $\pi:(M\setminus M')\rightarrow N$ \fi
            \ifshort $\pi$ \fi
            of the open items such that: (a) $\pi_i \neq \emptyset$ for at most $p$ agents; (b) $\gamma \cup \pi$ is envy-free?}\\
        \end{tabular}
      \end{tcolorbox}
      \end{center}     

While we formally study the decision variants of these problems for complexity-theoretic reasons, every algorithm obtained in this article is constructive and can also output a suitable allocation if one exists.

\iflong
\paragraph{\bf Parameterized complexity.}
We refer to the standard books for a basic overview of parameterized complexity theory~\cite{CyganFKLMPPS15,DowneyF13}.
At a high level, parameterized complexity studies the complexity of a problem with respect to its input size, $n$,  and the size of a parameter $k$. A problem is {\em fixed parameter tractable} by $k$, if it can be solved in time $f(k)\cdot \poly(n)$, where $f$ is a computable function. 
A less favorable, but still positive, outcome is an $\XP{}$ \emph{algorithm}, which has running-time $\bigoh(n^{f(k)})$; problems admitting such algorithms belong to the class $\XP$. Showing that a problem is \Wh[t] rules out the existence of a fixed-parameter algorithm under the well-established assumption that \Wh[t]$\neq \FPT$.
\fi

\section{Parameterizing by the Number of Open Items}

We start our investigation by considering the complexity of \textsc{Envy-Free Allocation Extension} when the number $k$ of open items is included in the parameterization. In other words, we ask under which conditions one can efficiently solve the extension problem for only a few open items. 

As our baseline, we observe that the problem admits a trivial \XP\ algorithm: one can enumerate all possible assignments of the open items to agents in $\bigoh{(n^k)}$ time and check whether any of these is envy-free. 
However, such algorithms are considered highly inefficient in the community~\cite{DowneyF13,CyganFKLMPPS15}, and the central question tackled by this section is whether (or under which conditions) one can achieve fixed-parameter tractability.

Before settling the problem when parameterized by $k$ alone, we first provide a fixed-parameter algorithm for \textsc{Envy-Free Allocation Extension} when the parameterization also includes the number of agent types. Intuitively, this provides a positive result for the case where there are only a few open items and the agents can be partitioned into a few groups with identical preferences. We remark that this setting is not trivial, as agents which we consider to have identical preferences (e.g., due to polling limitations) can and will often have different pre-assigned items.

The following observation will be useful in the proof and follows directly from the definition of agent types and envy-free allocations.

\begin{observation} \label{obs:agenttypes}
Assume we are given an envy-free partial allocation $\gamma$ such that no pair of agents in agent type $X$ envy each other. If a solution $\pi$ assigns a set $Q$  of positively valued items to an agent $j\in X$, then $\pi$ must assign items of the same value as $v_j(Q)$ to \emph{all} other agents in $X$.
\end{observation}

\kntfpt*
\ifshort
\begin{proof}[Proof Sketch]
As the very first step, we observe that if there is an agent $i$ envying agent $j$, then $\pi$ must assign some open item to $i$. Hence, we begin by exhaustively branching, i.e., guessing, for each such agent $i$, which of the at most $k$ items will be assigned to $i$, and restart our considerations with the instance updated accordingly. For the following, we hence consider that no agent envies another.

We begin by partitioning the set of agent types into \emph{small} and \emph{large} types. An agent type is classed as small when it contains at most $k$ agents. 
We define $Z$ as the set of agents in all of the small types, the size of $Z$ is upper-bounded by $k\cdot \agenttypes$.

Next, we exhaustively branch to determine the following information about the allocation of the $k$ items to agents. First of all, we branch over all partitions of the open items into bundles, where we will assume that $\pi$ assigns each of the bundles to distinct agents. Next, for each of the at most $k$ bundles, we branch to either determine which of the agents in $Z$ will receive it, or that it will not be assigned to any agent in $Z$. The overall branching factor up to this point is upper-bounded by $k^k\cdot (k\agenttypes+1)^k$. 

At this point, we check whether any agent from $Z$ envies another agent; if that is the case, then we can correctly reject the current branch, as we may assume to have precisely guessed the bundles assigned to all agents in $Z$. Similarly, no pair of agents outside of $Z$ may envy each other due to the exhaustive procedure carried out in the first paragraph of the proof. At this point, we recall that by Observation~\ref{obs:agenttypes} and the definition of large agent types, each agent outside of $Z$ must receive a bundle that they value as $0$. Hence, if an agent outside of $Z$ were to envy an agent in $Z$, we may also correctly reject the current branch.

At this point, it remains to assign the remaining open items to the agents outside of $Z$ without creating envy. For each of the bundles, we can determine whether assigning it to an agent $i \not \in Z$ creates envy from any other agent in the instance (where for agents in $Z$ we assume them to receive the bundles specified in our branching, while for agents outside of $Z$, we assume that nothing changed). To complete the proof, we construct an auxiliary bipartite graph $G$ where: 
\begin{itemize}
\item one side contains the set of all remaining bundles (i.e., those not assigned to $Z$), 
\item the other side contains the set of all agents outside of $Z$, and,
\item there is an edge between bundle $b$ and agent $i$ if and only if $b$ can be assigned to $i$ without creating envy from any other agent (i.e., bundles of value $0$).
\end{itemize}
\noindent
To complete the proof, it suffices to check whether $G$ admits a matching that saturates the set of all bundles. The overall running time can be upper-bounded by $(k\cdot \agenttypes)^{\bigoh(k)}\cdot n^2$.
\end{proof}
\fi

\iflong
\begin{proof}
As the very first step, we observe that if there is an agent $i$ envying agent $j$, then $\pi$ must assign some open item to $i$. Hence, we begin by exhaustively branching, i.e., guessing, for each such agent $i$, which of the at most $k$ items will be assigned to $i$, and restart our considerations with the instance updated accordingly. For the following, we hence consider that no agent envies another. If we can't do this, meaning $k$ reached $0$, the algorithm stops.

We begin by partitioning the set of agent types into \emph{small} and \emph{large} types. An agent type is classed as small when it contains at most $k$ agents. We define $Z$ as the set of agents in all of the small types, the size of $Z$ is upper-bounded by $k\cdot \agenttypes$.

Secondly, we exhaustively branch to determine the following information about the allocation of the $k$ items to agents. First of all, we branch over all partitions of the open items into bundles, of which there are $k^k$, where we will assume that $\pi$ assigns each of the bundles to distinct agents. 
Next, for each of the at most $k$ bundles, we branch to determine either: which of the agents in $Z$ will receive it; or that it will not be assigned to any agent in $Z$, i.e., the bundle goes to an agent from one of the large agent types. This gives us $k\cdot\agenttypes+1$ many possible agents to assign a bundle to. The overall branching factor up to this point is upper-bounded by $k^k\cdot (k\cdot\agenttypes+1)^k$.

At this point, we check whether any agent from $Z$ envies another agent; if that is the case, then we can correctly reject the current branch, as we may assume to have precisely guessed the bundles assigned to all agents in $Z$. Similarly, no pair of agents outside of $Z$ may envy each other due to the exhaustive procedure carried out in the first paragraph of the proof. At this point, we recall that by Observation~\ref{obs:agenttypes} and the definition of large agent types, each agent outside of $Z$ must receive a bundle that they value as $0$. Hence, if an agent outside of $Z$ were to envy an agent in $Z$, we may also correctly reject the current branch.
 
Now, it remains to assign the remaining open items to the agents outside of $Z$, the large agent types, without creating envy. For each of the bundles, we can determine whether assigning it to an agent $i \not \in Z$ creates envy from any other agent in the instance (where for agents in $Z$ we assume them to receive the bundles specified in our branching, while for agents outside of $Z$, we assume that nothing changed). To complete the proof, we construct an auxiliary bipartite graph $G$ where: 
\begin{itemize}
\item one side contains the set of all remaining bundles (i.e., those not assigned to agents in $Z$), 
\item the other side contains the set of all agents outside of $Z$, and, 
\item there is an edge between bundle $b$ and agent $i$ if and only if $b$ can be assigned to $i$ without creating envy from any other agent.
\end{itemize}
\noindent
We compute a maximum matching for this bipartite graph in polynomial time.

The matching outputs an allocation of bundles to agents.
It suffices to check whether $G$ admits a matching that saturates all bundles, i.e., all items have been assigned to some agent. 
Note that, if not all bundles are assigned, then there exists some bundle we couldn't assign to an agent without causing envy. 
Thus, if we did find such a matching, then there is an envy-free solution (by assigning the bundles according to the matching). 

Conversely, if there is a solution for \extensionproblem, these 3 steps are guaranteed to find it. This holds because in the case of the small agent types we enumerated their possible bundles, and there is no solution where the large agent types receive items that they value positively, because of Observation \ref{obs:agenttypes}. 
The overall running time can be upper-bounded by $(k\cdot \agenttypes)^{\bigoh(k)}\cdot n^2$.
\end{proof}
\fi

Turning back to \textsc{Envy-Free Allocation Extension} parameterized by $k$ alone, our next result shows that the aforementioned trivial \XP\ algorithm can be viewed as ``optimal'' in the sense of the problem not admitting any fixed-parameter algorithm under the \iflong well-established \fi complexity assumption of $\W[1]\not \subseteq \FPT$. Naturally, one would then ask whether fixed-parameter tractability can be achieved at least when the parameterization is enriched by the number $m_t$ of item types---a setting that can be seen as complementary to the one settled in Theorem~\ref{thm:kntfpt}. Surprisingly, we exclude fixed-parameter algorithms for the problem, even in this significantly more restrictive setting via a highly involved reduction. 

\kmtwhard*

\subsection{Proof of Theorem~\ref{thm:kmtwhard}}


\iflong
We reduce from the classical \W[1]-hard \multicolorclique problem: Given a $q$-partite graph, where each part is assigned a unique color, decide whether the graph contains a clique of size $q$. Let $\mathcal I = ((\text V,\text E),q)$ be an instance of \multicolorclique, where $q$ is the number of colors. Without loss of generality, assume that E only contains edges adjacent to vertices of different color and that for any pair of colors, there exist at least one pair of adjacent vertices (otherwise trivially a multicolored clique does not exist). 
Let $\text V = \bigcup_{1\leq k\leq q}\text V_k$, where $\text V_k$ is the set that contains the vertices of color $k$ and let $\text E = \bigcup_{1\leq i< j\leq q} \text E_{ij}$, where $\text E_{ij}$ is the set that contains the edges of $\text E$ that are adjacent to vertices that belong to sets $\text V_i$ and $\text V_j$ respectively. 
We will construct an instance $\mathcal I'$ of \extensionproblem\ and prove that an envy-free allocation extension for $\mathcal I'$ exists if and only if $\mathcal I$ contains a clique of size~$q$.

\paragraph{Construction.} For each vertex set $\text V_i$ we assume a \emph{vertex-agent group} $V_i = \{\alpha_j^i|~j\in \text V_i\}$ and for each edge set $\text E_{ij}$ an \emph{edge-agent group} $E_{ij} =\{\eta_z^{ij}|~z\in \text E_{ij}\}$. To avoid notational overflow, whenever it is clear from the context we omit the superscript in the notation of an agent. Moreover we assume that the agents in each group are provided in an arbitrary order, i.e., the agents in the vertex-agent group $V_i$ have the form $\alpha_1,\dots, \alpha_{|V_i|}$ and analogously the agents in the edge-agent group $E_{ij}$ have the form of $\eta_1,\dots, \eta_{|E_{ij}|}$. 

A notion that will be useful in the subsequent arguments is that of adjacent agents: 

\begin{definition}[Adjacent agents]
Let $z = (x,y)$ be an edge in $\mathcal I$ and let $\eta_z,\alpha_x\text{ and }  \alpha_y$ be the corresponding agents in $\mathcal I'$. We refer to any pair of these three agents as adjacent agents.
\end{definition}

For each vertex-agent group $V_i$, the construction will make use of pre-allocated items from the following three item types $\{\Box^i,\triangle^i, \bigstar^i\}$, while for each edge-agent group $E_{ij}$, the construction will also make use of pre-allocated items from the following three item types $\{\Box^{ij},\triangle^{ij}, \bigstar^{ij}\}$. The partial allocation $\gamma$ will provide the agents of the vertex-agent group $V_i$ with some combination of items from ``their'' types $\{\Box^i,\triangle^i, \bigstar^i\}$, while the agents of the edge-agent group $E_{ij}$ will receive items not only from ``their'' types $\{\Box^{ij},\triangle^{ij}, \bigstar^{ij}\}$ but also items of the types $\{\Box^i,\triangle^i,\Box^j,\triangle^j \}$. 

\paragraph{Partial allocation.}
We provide the details of how the partial allocation $\gamma$ is constructed. 

\begin{itemize}

\item Each agent $\alpha_x$ where $x\in[|V_i|]$, within a vertex-agent group $V_i$ holds a bundle consisting of:
\begin{itemize}
\item $x$ copies of $\Box^i$,
\item $2|V_i|^2- x^2-x$ copies of $\triangle^i$,
\item one copy of $\bigstar^{i}$.
\end{itemize}
Overall the bundle of the agent $\alpha_x$ in $V_i$ has the following form $\{\underbrace{\Box^i,\dots,\Box^i}_{x}, \underbrace{\triangle^i,\dots,\triangle^i}_{2|V_i|^2- x^2-x},\bigstar^i\}$.

\item Each agent $\eta_z$, within the edge-agent group $E_{ij}$ that is adjacent to some agents $\alpha_x\in V_i$ and $\alpha_y\in V_j$, holds a bundle consisting of
\begin{itemize}
\item $z$ copies of $\Box^{ij}$,
\item $2|E_{ij}|^2- z^2-z$ copies of $\triangle^{ij}$, 
\item one copy of $\bigstar^{ij}$,  
\item precisely the same number of $\Box^i$ ($\Box^j$) and $\triangle^i$ ($\triangle^j$) items as agent $\alpha_x$ ($\alpha_y$). That is $x$ copies of $\Box^i$, $y$ copies of $\Box^j$, $2|V_i|^2-x^2-x$ copies of $\triangle^i$ and $2|V_j|^2-y^2-y$ copies of $\triangle^j$.
\end{itemize}
Overall the bundle of an agent $\eta_z$ in $E_{ij}$ has the following form $$\{\underbrace{\Box^{ij},\dots,\Box^{ij}}_{z},\underbrace{\triangle^{ij},\dots \triangle^{ij}}_{2|E_{ij}|^2-z^2-z},\bigstar^{ij},\\ \underbrace{\Box^i,\dots,\Box^i}_{x},\underbrace{\triangle^i,\dots\triangle^i}_{2|V_i|^2-x^2-x},\underbrace{\Box^j\dots,\Box^j}_{y}, \underbrace{\triangle^j\dots,\triangle^j}_{2|V_j|^2-y^2-y}\}.$$
\end{itemize}

\noindent
We furthermore assume precisely  $q+\binom{q}{2}$ pairwise-distinct open items. A set $S = \{s_1,\dots, s_{q}\}$ of vertex items, containing one item tailored for each vertex-agent group (or equivalently color). A set $T = \{\tau_{12},\tau_{13},\cdots, \tau_{(q-1)(q)}\}$ of edge items, containing one item tailored for each edge-agent group $E_{ij}$.

Combining the open items with the pre-allocated items, we get that the total number of item types used in the generated instance $\mathcal{I'}$ is $4q+4\binom{q}{2}$.

\medskip
\noindent
\textbf{Agents' valuations.} We proceed to define the valuation function of the agents. Each agent will have a unique valuation function defined as follows:

\begin{itemize}
\item Each agent $\alpha_x$ within a vertex-agent group $V_i$ values:
\begin{itemize}
\item $\Box^i$ as $2x+1$,
\item $\triangle^i$ as $1$,
\item $\bigstar^{i}$ as $0$.
\end{itemize}
This describes the valuation for items allocated to $\alpha_x$ in $\gamma$. Next, for the item types that are not pre-assigned to $\alpha_x$, the agent uses the following valuation:
\begin{itemize}
\item Each $\triangle^\circ$ and each $\Box^\circ$ for $\circ\neq i$ has a value of $0$, 
\item For each $j\neq i$, $\bigstar^{ij}$ has a value of $0$,
\item Each $\bigstar^\circ$ whose valuation is not defined up to now, has a value that is precisely equal to the value of $\alpha_x$'s bundle in the partial allocation $\gamma$", i.e., $2|V_i|^2+x^2$ (see Lemma~\ref{lemma1}). 
\end{itemize}
We observe that the construction at this point ensures that agent $\alpha_x$ values his own bundle identically to that of any agent from a different vertex-agent group. Moreover, $\alpha_x$ values his own bundle identically to that of edge agents that are adjacent to $\alpha_x$, while for the remaining agents in those edge-agent groups, he prefers his own bundle.

Regarding the open items, $\alpha_x$ values $s_i$ as $1$ and for all $j\neq i$ the items $t_{ij}$ as $1$, and all the remaining open items as $0$.

\item Each agent $\eta_z$ within an edge-agent group $E_{ij}$ values:
\begin{itemize}
\item $\Box^{ij}$ as $2z+1$,
\item $\triangle^{ij}$ as $1$,
\item $\bigstar^{ij}$ as well as all remaining items included in its bundle by $\gamma$ as $0$.
\end{itemize}

Next, for the item types that are not pre-assigned to $\eta_z$, the agent uses the following valuation:
\begin{itemize}
\item each $\triangle^\circ$ and each $\Box^\circ$ for $\circ\neq ij$ has a value of $0$,
\item each $\bigstar^\circ$ where $\circ\neq ij$ has a value that is precisely equal to the value of $\eta_z$'s bundle as assigned by $\gamma$, i.e., $2|E_{ij}|+z^2$. 
\end{itemize}
\end{itemize}

We observe that by construction, $\eta_z$ will value its bundle exactly the same as that of every agent outside his group. Regarding the open items, $\eta_z$ values only the open item $t_{ij}$ as $1$ and all others as $0$.

\medskip
\noindent
We proceed by establishing two lemmas that must hold in any envy-free allocation extension. The statements of these lemmas arise from the way we defined the partial allocation $\gamma$ and set the agents' valuations.

First we show that we can allocate a single open item to an agent of some group, without creating envy among the agents within that group. We will prove that in the partial allocation $\gamma$, no envy exists among agents within the same agent group. In fact every agent strictly prefers its own bundle to the bundle of any other agent within the same group.

\begin{lemma}\label{lemma1}
Let $V_i$ be a vertex-agent group and consider any two vertex agents $\alpha_x$ and $\alpha_y$ within that group. It holds that $v_{\alpha_x}(\gamma_{\alpha_x}) > v_{\alpha_x}(\gamma_{\alpha_y})$.
 \end{lemma}

\begin{proof}
From the way we defined the partial allocation $\gamma$ and set the valuations, we get that: 
\begin{equation*}
\begin{split}
v_{\alpha_x}(\gamma_{\alpha_x}) &= x\cdot v_{\alpha_x}(\Box^i) + (2|V_i|^2-x^2-x)\cdot v_{\alpha_x}(\triangle^i)
+ v_{\alpha_x}(\bigstar^i)\\
&= x\cdot(2x+1) + (2|V_i|^2-x^2-x)\\
&= 2|V_i|^2 + x^2.
\end{split}
\end{equation*}
The valuation of agent $\alpha_x$ for the bundle of agent $\alpha_y$ in $\gamma$ is:
\begin{equation*}
\begin{split}
v_{\alpha_x}(\gamma_{\alpha_y}) &= y\cdot v_{\alpha_x}(\Box^i) + (2|V_i|^2-y^2-y)\cdot v_{\alpha_x}(\triangle^i) 
+ v_{\alpha_x}(\bigstar^i)\\ &= y\cdot(2x+1) + (2|V_i|^2-y^2-y)\\
&= 2|V_i|^2 +2xy -y^2.
\end{split}
\end{equation*}
Since the inequality $2|V_i|^2+x^2 > 2|V_i|^2 +2xy -y^2$ holds, we conclude that $v_{\alpha_x}(\gamma_{\alpha_x}) > v_{\alpha_x}(\gamma_{\alpha_y})$. Furthermore, since all valuations are integer values, even assigning one item that is valued 1 by every agent in a group, to some agent within that group, does not create envy among the agents in that group.
\end{proof}
Using the exact same arguments we can prove the following lemma for edge-agent groups.

\begin{lemma}\label{lemma2}
Let $E_{ij}$ be an edge-agent group and consider any two edge agents $\eta_x$ and $\eta_y$ within that group. It holds that $v_{\eta_x}(\gamma_{\eta_x}) > v_{\eta_x}(\gamma_{\eta_y})$.
 \end{lemma}

The next lemma establishes that for an extension for $\gamma$ to be envy-free, adjacent agents must receive triplets of open items that correspond to adjacent vertices in $\mathcal I$.

\begin{lemma}\label{lemma3}
In any envy-free allocation extension for $\gamma$, every triplet of open items $s_i,\tau_{ij}\text{ and }s_j$ must be allocated to adjacent agents belonging within the groups $V_i,E_{ij}\text{ and }V_j$ respectively.
\end{lemma}

\begin{proof}
First we prove that any open edge-item $\tau_{ij}$ must be allocated to some agent within the edge-agent group $E_{ij}$. If we allocate $\tau_{ij}$ to some agent within another edge-agent group say $E_{\kappa\lambda}$, then since all agents in $E_{ij}$ 
value item $\bigstar^{\kappa\lambda}$ (which appears in the bundle of every agent in $E_{\kappa\lambda}$) precisely equal to their bundle in $\gamma$, they become envious.
Similarly, if we allocate $\tau_{ij}$ to a vertex-agent group $V_i$, then since all agents in $E_{ij}$ value item $\bigstar^i$ (which appears in the bundle of every agent in $V_i$) equal to their bundle in $\gamma$, they all become envious. Combining this argument with the fact that edge-agents are solely interested in open edge-items, shows that the edge-item $\tau_{ij}$ must be allocated to some agent in the group $E_{ij}$.

Additionally, the open vertex-item $s_i$ must be allocated to an agent within the vertex-agent group $V_i$. If we allocate item $s_i$ to some agent within another vertex-agent group $V_j$, then from the way we defined the valuations each agent in $V_i$ values item $\bigstar^j$ (which is owned by every agent in $V_j$) precisely equal to their bundle as assigned in $\gamma$. Thus allocating $s_i$ to some agent within $V_j$ makes all agents in $V_i$ envious.
Similarly, if we allocate $s_i$ to some agent within the edge-agent group $E_{\kappa\lambda}$, then since all agents in $V_i$ value item $\bigstar^{\kappa\lambda}$ precisely equal to their bundle in $\gamma$ they all become envious. 

If $s_i$ is assigned to some agent $\eta_z$ within the edge-agent group $E_{ij}$, then there must exist an adjacent agent $\alpha_x$ that belongs to the vertex-agent group $V_i$, that becomes envious of $\eta_z$. That holds because $\alpha_x$ values the bundle of $\eta_z$ equally to his bundle in the partial allocation $\gamma$. Then, we cannot allocate any other open item in such a way to reduce this envy.

Finally we can show that every triplet of open items $s_i,\tau_{ij}\text{ and }s_j$ must be allocated to adjacent agents within the groups $V_i,E_{ij}\text{ and }V_j$ respectively. Indeed, from the above arguments we know that $\tau_{ij}$ will be allocated to some agent, say $\eta_z$, within $E_{ij}$. From the way we defined the valuations, only the two adjacent agents to $\eta_z$, that belong to vertex-agent groups $V_i$ and $V_j$ become envious, since both these agents value the bundle of $\eta_z$ precisely the same as their bundle in $\gamma$. Thus, to eliminate envy we must allocate the open items $s_i$ and $s_j$ to these agents respectively. Notice that from Lemmas \ref{lemma1} and \ref{lemma2} we do not induce envy among agents within the same group since we allocate only one item to each group.
\end{proof}

We proceed to show that a $q$-clique in the initial instance $\mathcal I$  exists if and only if an envy-free allocation extension for $\gamma$ exists in the constructed instance $\mathcal I'$.

 \paragraph{From $\mathcal I$ to $\mathcal I'$.} Assume that $\mathcal I$ contains a clique of size $q$. We extend the partial allocation $\gamma$ of $\mathcal I'$ by allocating the set of open items as follows: Pick an edge $z = (x,y)$ of the clique and without loss of generality assume that vertex $x$ has color $i$ and vertex $y$ has color $j$. Moreover let $\eta_z, \alpha_x\text{ and }\alpha_y$ be the corresponding adjacent agents in $\mathcal I'$. Allocate the triplet of items $s_i,\tau_{ij}\text{ and }s_j$ to the adjacent agents $\alpha_{x},\eta_z\text{ and },\alpha_y$ respectively. From Lemmas \ref{lemma1} and \ref{lemma2} we get that no envy emerges among agents of the same agent group. Additionally from Lemma \ref{lemma3}, we get that no envy emerges among agents of different agent groups. Consequently allocating the items that way produces an envy-free allocation extension for $\gamma$.
 
 \paragraph{From $\mathcal I'$ to $\mathcal I$.} Assume that the partial allocation $\gamma$ of $\mathcal I'$ extends to a complete envy-free allocation, say $\gamma'$. As argued above, the statement of Lemma \ref{lemma3} must be satisfied in order for $\gamma'$ to be envy-free. That means  that every triplet of open items $s_i, \tau_{ij}\text{ and } s_j$ is allocated to adjacent agents within the groups $V_i,E_{ij}\text{ and } V_j$ respectively. Moreover notice that only one agent within every vertex/edge-agent group receives an open vertex/edge item. Combining these facts, we get that any two agents of the vertex-agent groups that receive open items in $\gamma'$, must be adjacent agents. This implies that in $\mathcal I$, the vertices that correspond to the agents of the vertex-agent groups that receive the items are adjacent vertices of different color. Consequently these vertices form a multicolor clique of size $q$.
\fi

\ifshort
\begin{proof}[Proof Sketch]
We reduce from the classical \W[1]-hard \multicolorclique problem: given a $q$-partite (where each part is assigned a unique color) graph, decide whether the graph contains a clique of size $q$. Let $\mathcal I = ((\text V,\text E),q)$ be an instance of \multicolorclique, where $q$ is the number of colors and without loss of generality assume that $E$ contains only edges adjacent to vertices of different colors. 
Let $\text V = \bigcup_{1\leq k\leq q}\text V_k$, where $\text V_k$ is the set that contains all vertices of color $k$ and $\text E = \bigcup_{1\leq i< j\leq q} \text E_{ij}$, where $\text E_{ij}$ is the set that contains all edges of $\text E$ that are adjacent to vertices that belong to sets $\text V_i$ and $\text V_j$ respectively. 
We will construct an instance $\mathcal I'$ of \extensionproblem\ and prove that an envy-free extension for $\mathcal I'$ exists if and only if $\mathcal I$ contains a clique of size $q$.

\textbf{Construction:} For each vertex set $\text V_i$ we assume a \emph{vertex-agent group} $V_i = \{\alpha_j^i|j\in \text V_i\}$ and for each edge set $\text E_{ij}$ an \emph{edge-agent group} $E_{ij} =\{e_\ell^{ij}|\ell\in \text E_{ij}\}$. For each vertex-agent group $V_i$, the construction will make use of three pre-allocated item types $\{\Box^i,\triangle^i, \bigstar^i\}$. Each edge-agent group will also make use of three pre-allocated item types $\{\Box^{ij},\triangle^{ij}, \bigstar^{ij}\}$. The pre-allocation $\gamma$ will only provide agents in vertex-agent group $V_i$ with some combination of items from ``their'' types $\{\Box^i,\triangle^i, \bigstar^i\}$, while agents in edge-agent group $E_{ij}$ will receive items not only from ``their'' types $\{\Box^{ij},\triangle^{ij}, \bigstar^{ij}\}$ but also items of the types $\{\Box^i,\triangle^i,\Box^j,\triangle^j \}$. 

Below, we provide the details of how the partial allocation $\gamma$ is constructed. We assume that the agents in each group are provided in an arbitrary order, i.e., the agents in the vertex group $V_i$ have the form $\alpha_1,\dots, \alpha_{|V_i|}$ (and analogously for agents in edge groups).
\paragraph{Partial allocation $\gamma$.}
Each agent $\alpha_x$ within a vertex-agent group $V_i$ holds a bundle consisting of (1) one copy of $\bigstar^{i}$, (2) $x$ copies of $\Box^i$, and (3) $2|V_i|^2- x^2-x$ copies of $\triangle^i$.
Each agent $\eta_z$ within an edge-agent group $E_{ij}$ representing the edge between $\alpha_x\in V_i$ and $\alpha_y\in V_j$ holds a bundle consisting of (1) one copy of $\bigstar^{ij}$, (2) $z$ copies of $\Box^{ij}$, (3) $2|V_i|^2- z^2-z$ copies of $\triangle^{ij}$, and moreover (4) precisely the same number of $\Box^i$ ($\Box^j$) and $\triangle^i$ ($\triangle^j$) items as $\alpha_x$ ($\alpha_y$).

The instance will furthermore contain precisely $q+\binom{q}{2}$ pairwise-distinct open items: $s_1,\dots, s_q$ (one for each vertex agent group) and $t_{12}, t_{13}, \dots, t_{(q-1)q}$ (one for each edge agent group). This leaves the total number of item types at $4q+4\binom{q}{2}$. Intuitively, while the initial allocation $\gamma$ is envy-free, the constructed instance of \extensionproblem\ will force us to allocate the open items in a way where (1) each vertex and each edge agent group receives precisely one open item, and (2) to prevent envy, the recipient in the edge agent group must represent an edge such that both incident vertex-agents are recipients as well. However, to achieve this, we require very careful calibration of $\gamma$ and the specific valuations of the agents; in fact, each agent will have an entirely unique valuation function.

\paragraph{The valuation function.}
Each agent $\alpha_x$ within a vertex-agent group $V_i$ values: (1) $\triangle^i$ as $1$, (2) $\Box^i$ as $2x+1$, (3) $\bigstar^{i}$ as $0$. This describes the valuation for items allocated to $\alpha_x$ by $\gamma$. Next, for the item types that are not pre-assigned to $\alpha_x$, the agent uses the following valuation: (1) each $\triangle^\circ$ and each $\Box^\circ$ for $\circ\neq i$ has a value of $0$, (2) each $j\neq i$, $\bigstar^{ij}$ and/or $\bigstar^{ji}$ has a value of $0$, and (3) each $\bigstar^\circ$ whose valuation is not defined up to now has a value that is precisely equal to the value of $\alpha_x$'s bundle as assigned by $\gamma$.

We observe that the construction at this point ensures that $\alpha_x$ values its own bundle identically to that of any agent from a different vertex agent group. It is less obvious---but provable---that $\alpha_x$ strictly prefers their own bundle to the one of every other agent in their own vertex agent group. Moreover, $\alpha_x$ values its own bundle identically to that of edge agents that are incident to $\alpha_x$ in the graph, while for the other agents in those edge agent groups it prefers its own bundle. 

For the open items, $\alpha_x$ values $s_i$ as $1$, the items $t_{ij}$ and/or $t_{ji}$ for all $j\neq i$ as $1$, and all remaining open items as $0$. 

\item Each agent $\eta_z$ within an edge-agent group $E_{ij}$ values: (1) $\triangle^{ij}$ as $1$, (2) $\Box^{ij}$ as $2z+1$, (3) $\bigstar^{ij}$ as well as all remaining items included in its bundle by $\gamma$ as $0$. Next, for the item types that are not pre-assigned to $\eta_z$, the agent uses the following valuation: (1) each $\triangle^\circ$ and each $\Box^\circ$ for $\circ\neq ij$ has a value of $0$, and (2) each $\bigstar^\circ$ where $\circ\neq ij$ has a value that is precisely equal to the value of $\eta_z$'s bundle as assigned by $\gamma$.

We observe that by construction, $\eta_z$ will value its bundle exactly the same as as that of every agent outside of its own group. Moreover, similarly as before, it is possible to prove that $\eta_z$ strictly prefers their own bundle to the one of every other agent in their own edge agent group. Finally, for the open items, $\eta_z$ values its ``own'' open item $t_{ij}$ as $1$ and all others as $0$.

To complete the proof, two tasks remain. First, establish that the numbers are  indeed set up in a way to ensure that each agent strictly prefers their own bundle to that of other agents in their own group. The second is then to show that every $q$-clique in the initial graph corresponds to an envy-free allocation in the constructed instance following the intuition provided at the end of the partial allocation $\gamma$.
\qedhere
\end{proof}
\fi

\section{Parameterizing by the Number of Recipients}

In this section, we address the more complex situation of needing to allocate (a possibly large number of) open items to at most $p$ recipients. As mentioned in the Introduction, here it will be important to distinguish whether the set of these recipients is fixed and provided on the input (\REFAE), or whether the task also includes the identification of this set (\FEFAE). Moreover---and unlike in the case studied in the previous section---the existence of efficient algorithms will sometimes depend on whether we may assume the valuations to be encoded in unary, or whether they are encoded in binary. It will be useful to recall that lower bounds achieved in the former (latter) setting are called \emph{strong} (\emph{weak}).

As regards the complexity of these problems when parameterized by $p$ alone, it is easy to observe that both problems are weakly \NP-hard already when $p\leq 2$ as this directly generalizes both the previously-studied setting of all items being open~\cite{BliemBN16} and \textsc{Subset Sum}, both of which are weakly \NP-hard.
On the other hand, the problem is trivial for $p=1$. Below, we show that in the extension setting studied here, \NP-hardness holds for $p=2$ even in the unary (i.e., strong) setting, contrasting the known existence of an \XP algorithm for that case when all items are open~\cite{BliemBN16}.

\kappahardtwo*
\iflong
\begin{proof}

\fi
\ifshort
\begin{proof}[Proof Sketch]
\fi
We reduce from \indset that is known to be \NP-hard. An instance $\mathcal I$ of the \indset problem consists of a graph $(V,E)$ and an integer $\ell$. The task is to decide whether the graph contains a subset $V'\subseteq V$ where no two vertices in $V'$ are adjacent and $|V'| = \ell$. We construct an instance $\mathcal I'$ of \extensionproblem and prove that an envy-free allocation extension for $\mathcal I'$ with only two recipients exists if and only if instance $\mathcal I$ contains an independent set of size $\ell$.

\paragraph{Construction.} We assume a set $N = \{1,\cdots,|E|+1, |E|+2\}$ of $|E| + 2$ agents, consisting of a corresponding agent for each edge in $E$ and two additional agents labeled for simplicity as $|E|+1$ and $|E|+2$. We assume a set of items partitioned into the sets $G = \{g_1,\cdots,g_{|E|+1}, g_{|E|+2}\}$, that contains one item for each agent in $N$ and the set $\text{A} = \{a_1,\cdots, a_{|V|}\}$, that contains one item for each vertex of $V$. Finally, we assume the partial allocation $\gamma$, where for each item $g_i\in G$ and for each agent $i \in N$, $\gamma(g_i) =i$, i.e, agent $i$'s bundle in the partial allocation $\gamma$ is item $g_i$, while the items in $A$ are open and yet-to-be allocated.

The valuations for each agent $i \in N\setminus\{|E|+1,|E|+2\}$ are defined as:
\begin{itemize}
\item $v_i(g_i) = |V|$. 
\item $v_i(g_{|E|+1}) = |V|-1$ and $v_i(g_{|E|+2}) = 0$.
\item $v_i(a_j) = 1$, if the vertex $j$ that corresponds to item $a_j$ is adjacent to the edge represented by agent $i$, else 0.
\end{itemize}
The valuation for agent $|E|+1$, is defined as:
\begin{itemize}
\item $v_{|E|+1}(g_{|E|+1}) = |V|-2\ell$.
\item $v_{|E|+1}(g) = 0$, $\forall g\in G\setminus \{g_{|E|+1}\}$.
\item $v_{|E|+1}(a_j) = 1$, $\forall a_j\in A$.
\end{itemize}
Finally, the valuation for agent $|E|+2$ is defined as:
\begin{itemize}
\item $v_{|E|+2}(g_{|E|+2}) = |V|$ and $v_{|E|+2}(g_{|E|+1}) = 0$.
\item $v_{|E|+2}(g) = 2|V|-\ell$, for all $g\in G\setminus\{g_{|E|+1},g_{|E|+2}\}$.
\item $v_{|E|+2}(a_j) = 1$, $\forall a_j\in A$.
\end{itemize}

\ifshort
To complete the proof, it suffices to show that $(V,E)$ contains an independent set $Q$ of size $\ell$ if and only if allocating the open items corresponding to $Q$ to agent $|E|+1$ and all remaining open items to agent $|E|+2$ is a solution for $\mathcal{I}'$. To see that this holds, we note that if the items corresponding to both endpoints of an edge are allocated to the agent $|E|+1$, then $|E|+1$ would be envied by the corresponding edge agent. Moreover, if the number of items assigned to $|E|+2$ is less than $|V|-\ell$, then the agent $|E|+2$ would envy all of the edge-agents (possibly except for $|E|+1$). Finally, if the number of items assigned to $|E|+1$ is less than $\ell$, then $|E|+1$ would envy $|E|+2$ in view of the previous sentence.
\fi
\iflong
\paragraph{From $\mathcal I$ to $\mathcal I'$.} Assume that instance $\mathcal I$ contains an independent set of size $\ell$. Let $\text{IS}\subset A$ be the subset of open items $a_j\in A$, that correspond to vertices that belong to an independent set of $\mathcal I$. We will show that by allocating bundle $\text{IS}$ to agent $|E|+1$ and bundle $A\setminus \text{IS}$ to agent $|E|+2$, we get an envy-free allocation extension for $\mathcal I'$. Thus, in the complete allocation $\gamma' = \gamma\cup \pi$, it holds that $\gamma'(g_{|E|+1}\cup \text{IS}) = |E|+1$ and $\gamma'(g_{|E|+2}\cup (\text{A}\setminus \text{IS})) = |E|+2$.

We will argue that in the complete allocation $\gamma'$, no agent envies another agent. For agent $|E|+1$ it holds that $v_{|E|+1}(g_{|E|+1}\cup \text{IS}) = v_{|E|+1}(g_{|E|+1})+ v_{|E|+1}(\text{IS}) = |V|-2\ell + \ell = |V|-\ell$, which holds because we assume that valuations are additive, agent $|E|+1$ values any item in $\text{IS}$ equal to 1, and by the hypothesis each item in IS corresponds to some vertex in the independent set of size $\ell$.  

Agent $|E|+1$ does not envy any agent in $N\setminus \{|E|+2\}$, since he is indifferent towards items in $G\setminus \{g_{|E|+1}\}$. Finally, the valuation of $|E|+1$ for the bundle of agent $|E|+2$ is $v_{|E|+1}(g_{|E|+2}\cup (\text{A}\setminus \text{IS})) = v_{|E|+1}(g_{|E|+2})+v_{|E|+1}(A\setminus \text{IS}) = 0 + |V|- \ell = |V|-\ell$. Consequently, $|E|+1$ does not envy agent $|E|+2$ either. 

Agent's $|E|+2$ valuation for his bundle is $v_{|E|+2}(g_{|E|+2}\cup (\text{A} \setminus \text{IS})) = v_{|E|+2}(g_{|E|+2}) + v_{|E|+2}(\text{A} \setminus \text{IS}) = |V|+|V|-\ell = 2|V|-\ell$, since the valuations are additive and there exist $|V|-\ell$ items in the bundle $\text{A}\setminus \text{IS}$, each valued 1. There exist no envy towards agent $|E|+1$ since $v_{|E|+2}(g_{|E|+1}\cup \text{IS}) = v_{|E|+2}(g_{|E|+1}) + v_{|E|+2}(\text{IS}) = \ell < 2|V|-\ell$. Finally agent $|E|+2$ values any item in $G\setminus \{g_{|E|+1},g_{|E|+2}\}$ by $2|V|-\ell$, thus he does not envy any agent in $N\setminus \{|E|+1,|E|+2\}$.  

None of the agents $i\in N\setminus \{|E|+1,|E|+2\}$, receive an open item from $\text{A}$, while they value their bundle by $v_i(g_i) = |V|$. Notice that no two items in $\text{IS}$ can correspond to adjacent vertices of $\mathcal I$. This holds from the way IS is defined and the assumption that $\mathcal I$ contains an independent set of size $\ell$. Consequently the valuation of any agent $i\in N\setminus \{|E|+1,|E|+2\} $ for \text{IS} is at most 1. Thus their valuation for the bundle of agent $|E|+1$ is $v_i(g_{|E|+1}\cup \text{IS}) = v_i(g_{|E|+1}) + v_i(\text{IS}) \leq  |V|$, meaning that no envy exists towards agent $|E|+1$. 

Finally no agent $i\in N\setminus \{|E|+1,|E|+2\}$ is interested in item $g_{|E|+2}$, while $i's$ value for the bundle of agent $|E|+2$ in $\gamma'$ is $v_i(g_{|E|+2}\cup (\text{A}\setminus \text{IS})) \leq 2 \leq |V|$, meaning that no envy towards agent $|E|+2$ exists.

The above arguments show that by allocating the bundle IS to agent $|E|+1$ and the bundle A$\setminus$IS to agent $|E|+2$, we get an envy-free allocation extension $\gamma'$ of the partial allocation $\gamma$ for the instance $\mathcal I'$.

\paragraph{From $\mathcal I'$ to $\mathcal I$.} Assume that there exists an envy-free allocation extension $\gamma'$ of $\gamma$ for $\mathcal I'$. First, we will argue that in any such extension, the open items in $\text{A}$ are allocated only to the agents $|E|+1$ and $|E|+2$. In the partial allocation $\gamma$, it holds that  $v_{|E|+2}(g_{|E|+2}) = |V|$, but agent $|E|+2$ envies all agents $i\in N\setminus \{|E|+1,|E|+2\}$, since $v_{|E|+2}(g_i) = 2|V|-\ell$. Consequently, since agent $|E|+2$ values any open item by 1, in allocation $\pi$ he must receive at least $|V|-\ell$ open items. In turn agent $|E|+1$ has a valuation of $v_{|E|+1}(g_{|E|+1}) = |V|-2\ell$ and the only way not to envy agent $|E|+2$ is to get exactly $\ell$ open items. Otherwise, since $|E|+1$ values any open item by 1 and agent $|E|+2$ must take at least $|V|-\ell$ open items if agent $|E|+1$ is allocated less than $\ell$ open items, then the valuation he receives will be lower than $|V|-\ell$ which is the lower bound of his valuation for the bundle of agent $|E|+2$ in any envy-free extension. Consequently in $\pi$, agents $|E|+1$ and $|E|+2$ must receive exactly $\ell$ and $|V|-\ell$ respectively open items from set A. 

Let $A_{|E|+1}$ be the subset of $\ell$ items of A that are allocated to agent $|E|+1$ and $A_{|E|+2}$ be the subset of $|V|-\ell$ items of A that are allocated to agent $|E|+2$ in allocation $\pi$. Notice that no agent $i\in N\setminus \{|E|+1,|E|+2\}$ envies agent $|E|+2$ since $v_i(g_{|E|+2}\cup A_{|E|+2}) \leq 2\leq |V| = v_i(g_i)$. 
Moreover, since no agent can be envious in $\pi$, that must mean that for each $i\in N\setminus\{|E|+1,|E|+2\}$, $v_i(g_{|E|+1}\cup A_{|E|+1}) \leq v_i(g_i)$ which consequently means that $v_i(A_{|E|+1})\leq 1$. This can hold only if no two items in $A_{|E|+1}$ correspond to vertices that are adjacent to the same edge in $\mathcal I$. Consequently, since $\gamma'$ is an envy-free allocation extension of the partial allocation $\gamma$ for $\mathcal I'$, it must be the case that the items in the subset $A_{|E|+1}$ correspond to vertices of instance $\mathcal I$ that form an independent set of size $\ell$.
\fi
\end{proof}

We remark that Theorem~\ref{thm:kappahardtwo} is tight in the sense that both problems are also in \NP. Indeed, in both cases one can verify whether a provided complete allocation is envy-free in polynomial time.
Having settled the intractability of \REFAE and \FEFAE w.r.t.\ $p$ alone, we now ask whether one can achieve tractability at least in settings with a small number of agent or item types. 

We begin by considering the former parameterization, which was the one that yielded fixed-parameter tractability for the setting investigated in the previous section. Here, we can immediately exclude fixed-parameter tractability for all of the problem variants considered in this section. This holds due to the aforementioned fact that finding an envy-free allocation of a set of open items to agents is known to be strongly \W[1]-hard and weakly para\NP-hard when parameterized by the number of agents~\cite{BliemBN16}. Below, we at least show that both problems admit \XP algorithms in the unary-valuation case, complementing the former lower bound.


\kappantxp*


\ifshort
\begin{proof}[Proof Sketch]
For \FEFAE, we start by branching to determine which of the at most $p$ agents will be the recipients, requiring a branching factor of at most $n^p$. For \REFAE, we skip this step.

At this point, we now perform a \iflong fairly standard \fi dynamic programming subroutine for an input instance $\mathcal{I}$ and then argue that this subroutine, in fact, solves our problem. For the subroutine, we construct a table which stores the following information: for each recipient $i\in [p]$ and each agent type $Z$, we store the valuation of the bundle assigned to $i$ from the perspective of agents of type $Z$. We note that each table entry consists of $p\cdot \agenttypes$ unary-encoded values. We assume that the items are processed in an arbitrary but fixed order, and for each item we update the table by exhaustively assigning it to each of the possible recipients and updating the table accordingly. 

At the end of this subroutine, we loop through each of the at most $|\mathcal{I}|^{\bigoh({\agenttypes\cdot p})}$ entries in the table, and notice that each such entry provides us with complete information about how each agent values each bundle in $\mathcal{I}$ in all possible assignments corresponding to that table entry; hence, it suffices to check, in polynomial time, whether at least one such table entry results in an envy-free assignment in $\mathcal{I}$. 
\end{proof}
\fi

\iflong
\begin{proof}

Firstly, recall that for \REFAE we know exactly which of the agents are the $p$ recipients. However, for \FEFAE, we first start by branching to determine which of the at most $p$ agents will be the recipients, requiring a branching factor of at most $n^p$. 

Having the set of recipients as part of the input, we denote the set of recipients as $[p]$.
Let us fix some arbitrary ordering $1, \dots, m$ in which we will assign items in $A$ (the set of open items).
We now construct a table $T_t$ for $t\in [A]$ to store how much agents envy one another under different configurations. 
A configuration $x$ at step $t$ is a possible way of giving out the first $t$ items to the recipients. Each entry $t$ has a vector $c^x_t$, the configuration $x$ of $t$ many items, for each configuration $x$. 

For each configuration $x$, we store for each recipient $i\in [p]$ and each agent type $Z$ the valuation of the bundle assigned to $i$ from the perspective of agents of type $Z$, denoted as $v^Z_i$. Observe that this includes the valuation agent $i$ has for their bundle. 
It is sufficient to compute valuations from each agent type $Z$ against each of the $p$ recipient agents because all agents in $Z$ see the bundles of the recipient agents the same, and the bundles of the non-recipient agents won't change.
We note that each table entry consists of $p\cdot \agenttypes$ unary-encoded values.

We have just one vector for the starting table entry $T_0$, as none of $A$ has been assigned yet. Hence, this is exactly the valuations $V^Z_i$ for in the partial allocation $\gamma$ (which is EF). 
At $T_1$ we now start creating configurations of items and then computing the valuations from each agent type $Z$ towards each recipient agent $i \in [p]$. Note also that in $T_1$ we will have $p$ many vectors; we are assigning the first item to each of the possible recipients. Additionally observe that we are not storing the configurations (how we allocate the items), just the valuations induced by the configuration. Hence, we can also discard duplicate entries.

Going forward, we compute each entry $T_{t+1}$ by giving out the next item, the $t+1$th item, to each of the recipients and updating the table accordingly. Crucially, each entry will have $n_t \cdot p$ many values (the number of valuations we compute), and so---since the valuations are encoded in unary---after assigning all open items, the table will be storing at most $|\mathcal{I}|^{\agenttypes \cdot p}$ many entries. This yields a total running time of $|\mathcal{I}|^{\bigoh({\agenttypes\cdot p})}\cdot |\mathcal{I}|$.\

At the end of this subroutine, we loop through each of the at most $|\mathcal{I}|^{\bigoh({\agenttypes\cdot p})}$ entries in the table, and notice that each such entry provides us with complete information about how each agent values each bundle in $\mathcal{I}$ in all possible assignments corresponding to that table entry. Hence, it suffices to check, in polynomial time, whether at least one such table entry results in an envy-free assignment in $\mathcal{I}$. 
If there does not exist such an entry in the final table, we can conclude there is no solution because this means among all configurations of the open items no extension exists which is EF.
\end{proof}
\fi


Surprisingly, when dealing with a small number of recipients, we prove that it is the number $\itemtypes$ of item types which yields better tractability results for the considered problems---a situation that is entirely opposite to that of the previous section. In particular, when parameterizing by $p+\agenttypes$ we obtain a fixed-parameter algorithm for \REFAE and an \XP algorithm for \FEFAE. We remark that the latter result is the best one could have at this point hoped for due to Theorem~\ref{thm:kmtwhard} ruling out fixed-parameter algorithms even in a strictly more restrictive setting.

\kappamtfpt*
    

\ifshort
\begin{proof}[Proof Sketch]
For \FEFAE, we start again by branching to determine which of the at most $p$ agents will be the recipients, requiring a branching factor of at most $n^p$. At this point, we can solve both problems via an encoding into an 
\textsc{ILP} 
using a number of variables that depend purely on the parameters. In particular, for each of the at most $p$ recipients and each of the $\itemtypes$ item types, we construct a dedicated variable which will determine precisely how many items of this type will be assigned to the given recipient. The valuation of each agent can then be easily captured as a linear expression of the variables, and the requirement that the final allocation is envy-free can be captured via linear constraints. Since the number of variables is upper-bounded by $p\cdot \itemtypes$, we can solve the resulting instance in time $(p\cdot \itemtypes)^{\bigoh(p\cdot \itemtypes)}\cdot |\mathcal{I}|^{\bigoh(1)}$~\cite{Lenstra83,Dadush12}.
\end{proof}
\fi

\iflong
\begin{proof}
For \FEFAE, we start by branching on which agents will receive items; this requires at most $n^p$ branches.
Let us fix some set of recipient agents $[p]$.
We skip this first step for \REFAE. The remainder of the proof is the same for both \REFAE and \FEFAE.

We now define an \textsc{Integer Linear Programming} (ILP) instance 
such that the ILP instance is satisfiable if and only if there exists an envy-free complete allocation consistent with $\gamma$. 

We define a variable $x^t_i$ for each agent $i \in [p]$ and each item type $t \in T$ which denotes the number of items of type $t$ that agent $i$ receives. We first make sure that any satisfying assignment for ILP assigns the correct number of items of each type by including for each type $t$ the equality (encoded as two inequalities) 
\begin{equation}
    \sum_{i\in [p]}x^t_i = |t|,\label{eq:numberItems}
\end{equation} 
where $|t|$ denotes the number of items of type $|t|$.


In addition, we require that no pair of agents that receive items envy each other after the extension. Formally, for all pairs of agents $i \in [p]$ and $j \in [p]$ we add the following inequality to the ILP instance
\begin{equation}
v_i(\pi_i) + \sum_{t \in T} x^t_i \cdot v_i(t) \geq v_i(\pi_j) +\sum_{t \in \setitemtypes} v_i(t) \cdot x^t_j,\label{eq:envyRecepients}
\end{equation}
This inequality says that agent $i$ is not envious of $j$ after the extension; i.e., $i$ getting $x^t_i$ many of type $t$ and $j$ getting $x^t_j$ many does not create envy from $i$ to $j$. 

We also create an inequality which handles agents not receiving any items (agents not in $p$) to ensure they do not become envious. In particular, for $j \notin [p]$ we construct the following inequality:
\begin{equation}
v_i(\pi_j) \geq 
v_i(\pi_i) + \sum_{t \in T} x^t_i \cdot v_i(t) .\label{eq:envyNonRecepients}
\end{equation}

Observe, the number of variables is upper-bounded by $p\cdot \itemtypes$. Hence, we can solve the resulting instance in time $(p\cdot \itemtypes)^{\bigoh(p\cdot \itemtypes)}\cdot |\mathcal{I}|^{\bigoh(1)}$~\cite{Lenstra83,Dadush12}. Thus, this is FPT-time for \REFAE and XP-time for \FEFAE as we have $\bigoh(n^p)$ different ILP instances to check.

To complete the proof, it remains to argue that we have an envy-free solution if and only if we have a solution to the constructed ILP. 

Assume we have a solution for the ILP, this solution describes to us what agents will receive which bundles. More precisely, for each agent $i\in [p]$ and item type $t\in T$ we assign arbitrary $x_i^t$ many items of type $t$ to agent $i$. First note that  Equality~\eqref{eq:numberItems} ensures that such an assignment is possible and all items have been assigned to some agent. 
Assume that this solution is not an envy-free extension to $\gamma$. This means that there must exist a pair of agents $i$ and $j$ such that $i$ envies $j$. If $i$ is a recipient agent and $j$ is also a recipient, this would mean that Inequality~\eqref{eq:envyRecepients} 
does not hold for the pair $i,j$. On the other hand, if $j$ was not a recipient, it is not possible $i$ envies $j$ because $\gamma$, the partial allocation, is envy-free, all items have non-negative values for all agents, and $j$ has not received anything in the extension.
If agent $i$ was not a recipient and $j$ is a recipient, this means Inequality~\eqref{eq:envyNonRecepients} would not hold for this pair. Thus, the solution for the ILP must be an envy-free solution.

Now, given an solution to \extensionproblem we will argue that this is also gives us a solution to the ILP. To obtain the solution for the ILP, for every agent $i\in [p]$ and every item type $t\in T$, we let $x_i^t$ be the number of items of type $t$ the agent $i$ received. Equality~\eqref{eq:numberItems} holds for each item type $t\in T$, because all items of the type have been assigned. Inequalities~\eqref{eq:envyRecepients}~and~\eqref{eq:envyNonRecepients} directly describe that the agent $i$ does not envy the agent $j$ for the particular pair of agents $i,j$, and so all of these inequalities have to hold as well.
This concludes the proof.
\end{proof}
\fi




\section{Extending Allocations Beyond Envy-Freeness}
While it seems that the algorithmic upper and lower bounds \iflong presented in the main contribution of our article \fi could be translated also to the analogous problems of extending EF1 and EFX allocations, in this section we turn towards a different, highly studied aspect of such allocations---specifically, the question of whether such allocations are guaranteed to exist~\cite{budish2011combinatorial,CaragiannisKMPS19}. We show that in the extension setting, this question can be completely resolved based on the level of fairness we assume in the partial allocation.

First of all, it is obvious that if no fairness guarantees are provided for the partial allocation, then one cannot hope to guarantee an extension to an EF1 or EFX allocation (since, e.g., there might not be any open items left at all). On the other hand, if we assume the partial allocation to be envy-free, we obtain the following.

\iflong
\begin{proposition}
\fi
\ifshort
\begin{proposition}
\fi
Every envy-free partial allocation \iflong of items to agents \fi can be extended to an \textup{EF1} allocation. At the same time, there exists an envy-free partial allocation of items to $2$ agents which cannot be extended to an \textup{EFX} allocation.
\end{proposition}

\ifshort
\begin{proof}[Proof Sketch]
To extend an envy-free allocation to EF1, we show that it is possible to compute an allocation $\pi$ of open items such that the final allocation consisting of $\pi$ and the partial allocation $\gamma$ is EF1. This is done by executing a standard algorithm for computing EF1 allocations (Envy Cycle Elimination~\cite{lipton2004approximately}) on top of the partial allocation $\gamma$.
For the second part of the claim, consider two agents with identical valuations such that $\gamma$ assigns each an item that each of them values $1$. There are three open items to be allocated that the agents value 3, 4, and 9 respectively. 
\end{proof}
\fi
\iflong
\begin{proof}
The first part of the claim follows almost immediately from the fact that the utilities are additive. 
Let $\gamma$ be an envy-free allocation.  This means that $v_i(\gamma_i) \geq v_i(\gamma_j)$ for every pair of agents $i$ and $j$.
We compute a partial allocation $\pi$ that is EF1 employing a standard algorithm, like Envy Cycle Elimination~\cite{lipton2004approximately}, or Round Robin. 
Then, since $\pi$ is EF1 it holds that $v_i(\pi_i) \geq v_i(\pi_j) - \max_{a \in \pi_j} v_i(a)$, for every pair of agents $i$ and $j$. 
Thus, if we combine all the above we get that $v_i(\gamma_i \cup \pi_i) = v_i(\gamma_i) + v_i(\pi_i) \geq v_i(\gamma_j)+v_i(\pi_j) - \max_{a \in \pi_j} v_i(a)$, for every pair of agents $i$ and $j$. Thus, the extended allocation $\gamma \cup \pi$ is EF1.

For the second part of the claim consider two agents with identical valuations. Under $\gamma$, each agent gets an item that he values 1 and there is one open item that each agent values 2. No matter which agent agent receives the open item, the other agent will still envy him after removing the item of value 1. Hence, both possible extended allocations do not satisfy the conditions of EFX.
%
\end{proof}
\fi

Interestingly, the latter non-existence result contrasts with the known fact that EFX allocations from scratch always exist for $2$ agents~\cite{plaut2020almost}. 

Moreover, rather than starting with an envy-free partial allocation, we can also consider some other ``reasonable'' partial allocations. For example, if we start with a partial allocation which is EF1 and the envy-free graph is acyclic~\cite{lipton2004approximately}, there is a polynomial-time algorithm to find an EF1 complete allocation~\cite{igarashi2024fair}. On another note, it is, in fact, \NPc to find a complete allocation which is EF1 and PO, by starting from an EF1 partial allocation even when agents have identical valuations~\cite{igarashi2024fair}. 

The final remaining question is whether a partial EFX allocation extends to a complete EF1 allocation. We resolve this in the negative below.

\iflong
\begin{proposition}
\fi
\ifshort
\begin{proposition}
\fi
There exists an \textup{EFX} partial allocation \iflong of items \fi to $2$ agents which cannot be extended to an \textup{EF1} allocation.
\end{proposition}
\ifshort
\begin{proof}[Proof Sketch]
Consider the following instance with two agents and three items $x,y,z$. 
The valuations of the agents are as follows:
$v_1(x)=10, v_1(y)=0, v_1(z)=1$; 
$v_2(x)=0, v_2(y)=10, v_2(z)=1$.
Consider now the partial allocation that gives $x$ to agent 2 and $y$ to agent 1. Trivially, this is an EFX allocation since envy can be eliminated by removing the only item the other agent received. 
Now, no matter which agent receives item $z$, the constructed solution will not be EF1.
\end{proof}
\fi

\iflong
\begin{proof}
Consider the following instance with two agents and three items $x,y,z$. 
The valuations of the agents are as follows:
$v_1(x)=10, v_1(y)=0, v_1(z)=1$; 
$v_2(x)=0, v_2(y)=10, v_2(z)=1$.
Consider now the partial allocation that gives $x$ to agent 2 and $y$ to agent 1. Trivially, this is an EFX allocation since envy can be eliminated by removing the only item the other agent receives. 
Now, no matter which agent receives item $z$ the constructed solution will not be EF1. Indeed, assume that agent 1 gets $z$. Observe that the value of agent 2 for $\{y,z \}$ is $10+1$, while for their bundle is 0. Hence after the removal of any item from the bundle of agent 1, agent 2 will still have a positive value for it and thus the envy is not eliminated.
\end{proof}
\fi

\section{Concluding Remarks}

Our paper initiates the study of fairly extending partial allocations of indivisible items and the frontiers of tractability for this problem. The presented results showcase that the complexity of this task varies as it strongly depends on the chosen parameters---with the number of agent types and item types each leading to tractability in different settings. This naturally gives rise to the question of what other parameterizations can yield fixed-parameter algorithms for this natural but previously overlooked problem. For instance, can one obtain such algorithms when simultaneously parameterizing by the number \itemtypes\ of item types {\em and} the number \agenttypes\ of agent types? 
If this combination of parameters is not sufficient for tractability, can the inclusion of the number of recipients, $p$, in addition to \agenttypes and \itemtypes lead to tractability?

A different direction that has proven fruitful recently, is to consider the problem when agents form a social network~\cite{BredereckKN22,EibenGHO23}. In this setting, each agent compares their bundle against the bundles of their ``friends''. What are the graph structures on the social network that make the problem tractable?

\section*{Acknowledgments}
Argyrios Deligkas and Stavros D. Ioannidis acknowledge the support of the EPSRC grant EP/X039862/1.
Robert Ganian acknowledges support from the Austrian Science Fund (FWF, projects 10.55776/Y1329 and 10.55776/COE12) and the Vienna Science Foundation (WWTF, project 10.47379/ICT22029).

\bibliographystyle{alpha}
\bibliography{references}

\end{document}